\documentclass[twocolumn]{article}

\usepackage[a4paper, total={6.5in, 9.5in}]{geometry}

\usepackage{cite}
\usepackage{amsmath,amssymb,amsfonts, mathtools}
\usepackage{bbm, xfrac}
\usepackage{algorithm, algorithmic}

\usepackage{graphicx}
\usepackage{tikz}
\usetikzlibrary{matrix,decorations.pathreplacing, calc, positioning,fit, decorations.pathreplacing,calligraphy}

\usepackage{subcaption, textcomp, placeins}

\usepackage{amsthm}
\theoremstyle{plain}
\newtheorem{theorem}{Theorem}
\newtheorem{lemma}{Lemma}

\newtheorem{corollary}{Corollary}[theorem]

\theoremstyle{definition}
\newtheorem*{definition}{Definition}
\theoremstyle{definition}
\newtheorem{assumption}{Assumption}
\theoremstyle{remark}

\def\BibTeX{{\rm B\kern-.05em{\sc i\kern-.025em b}\kern-.08em
    T\kern-.1667em\lower.7ex\hbox{E}\kern-.125emX}}

\usepackage{titlesec, abstract}
\titleformat{\section}{\large\bfseries}{\thesection}{0.5em}{}
\titleformat{\subsection}{\large\itshape}{\thesubsection}{0.5em}{}

\begin{document}
\title{\textbf{Exploiting Sparsity for Localization of Large-Scale Wireless Sensor Networks}}
\author{Shiraz Khan, Inseok Hwang, and James Goppert\thanks{S. Khan, I. Hwang and J. Goppert are with the School of Aeronautics and Astronautics, Purdue University, West Lafayette, IN 47907, USA (e-mail:shiraz@purdue.edu; ihwang@purdue.edu; jgoppert@purdue.edu)}
}
\date{\vspace{-5ex}}

\twocolumn[
  \begin{@twocolumnfalse}
\maketitle
\begin{abstract}
Wireless Sensor Network (WSN) localization refers to the problem of determining the position of each of the agents in a WSN using noisy measurement information. In many cases, such as in distance and bearing-based localization, the measurement model is a nonlinear function of the agents' positions, leading to pairwise interconnections between the agents. As the optimal solution for the WSN localization problem is known to be computationally expensive in these cases, an efficient approximation is desired. In this paper, we show that the inherent sparsity in this problem can be exploited to greatly reduce the computational effort of using an Extended Kalman Filter (EKF) for large-scale WSN localization. In the proposed method, which we call the Low-Bandwidth Extended Kalman Filter (LB-EKF), the measurement information matrix is converted into a banded matrix by relabeling (permuting the order of) the vertices of the graph. Using a combination of theoretical analysis and numerical simulations, it is shown that typical WSN configurations (which can be modeled as random geometric graphs) can be localized in a scalable manner using the proposed LB-EKF approach.

\vspace{4pt}
\end{abstract}
\hfill
\end{@twocolumnfalse}
]
\saythanks
\section{Introduction}
\label{sec:introduction}
Wireless Sensor Networks (WSNs) are interconnected agents that sense and navigate their environment, where the interconnections can model various aspects of the WSN, such as the ability of two agents to sense, influence, or communicate with each other. Many real-world systems can be abstracted as WSNs, including groups of autonomous vehicles and fleets of ground-based rovers.
WSN localization refers to the problem of accurately estimating the position of each of the agents by processing noisy measurement data. 

In many WSN applications, the agents are able to obtain relative measurements between their states, such as in the case of inter-agent distance measurements (e.g., using time difference of arrival) or relative bearing measurements (e.g., using cameras) \cite{iet_loc_distance, iet_relative_bearing, zhao2019bearing}. Moreover, if the agents are equipped with communication capabilities, the received signal strength (RSS) values of the communication links can be used as relative measurements, which is a low-cost approach for WSN localization \cite{RSS2019efficient}. 
Relative measurements are especially useful when the agents are in a GNSS-denied area, or when some of them are being subject to GNSS-spoofing attacks. In such cases the relative measurements can serve as redundant, reliable sources of information that can localize the WSN as well as help to detect sensor faults and spoofing attacks. 
Assuming that these relative measurements are available between any two agents which are within some given distance of each other, the WSN can be represented using a geometric graph (also called a disk graph) \cite{bdo_localization2007}. 
The problem of WSN localization on geometric graphs using relative measurements has been shown to be NP-hard \cite{aspnes2006theory}. Thus, much academic effort has gone into developing efficient state estimation algorithms which can solve this problem approximately.

Classical approaches for WSN localization like maximum likelihood estimation require a semidefinite programming relaxation before they can be implemented efficiently \cite{sdp2018}. A drawback of classical state estimation algorithms is that they only track the state estimate, whereas Bayesian estimators (such as those based on Kalman filtering theory) estimate the covariance of the estimation error as well. The error covariance is valuable for control and decision applications, such as checking for sensor faults using a residual test \cite{khan2022robust}. Another advantage of Kalman filter-based approaches is that they have provable convergence guarantees, which is not the case for classical estimators that only guarantee the asymptotic consistency of the estimate \cite{ekf1995, ekf1999}. The Extended Kalman Filter (EKF) and Unscented Kalman Filter (UKF) algorithms are two of the most commonly used Kalman filter-based approaches which can accommodate nonlinear measurement models. 
While there has been extensive research in the application of EKF to problems having a low-dimensional state vector, e.g., using consensus or diffusion-based strategies \cite{consensus}, these algorithms are not applicable for large-scale WSN localization, which has a high-dimensional state vector. Moreover, estimation algorithms for low-dimensional problems have been made computationally efficient by specializing to geometric graphs \cite{dimakis2006geographic}. On the other hand, EKF-based approaches for WSN localization in the literature do not exploit graph theory \cite{2002ekf, 2019ekf}. The authors of \cite{moura2008distributing} note that the computational complexity of a large-scale implementation of the EKF is governed by the error covariance update step, which involves large matrix inversion operations. The complexity of the covariance update step can be reduced by exploiting the sparsity of the matrices involved in its computation. Specifically, matrices having a small bandwidth (i.e., having all of their non-zero elements close to the main diagonal) are highly efficient for scientific computing and memory access \cite{kavcic2000matrices, liao2016improved_band_matrix_computation, rcm2017}. In \cite{moura2008distributing}, the authors design a Kalman filter for large-scale applications by assuming that the measurement information matrix has a small bandwidth. 
However, the foregoing assumption is not satisfied in general, so the algorithm in \cite{moura2008distributing} is not suitable for large-scale WSN localization.

In this paper, we develop a novel approach for localizing large-scale WSNs in a scalable manner by leveraging the inherent sparsity of the problem. 
In the proposed approach, called the $L$-Banded Extended Kalman Filter (LB-EKF),
the sparsity pattern of the measurement information matrix is modified by relabeling (or reordering) the vertices of the graph. We show that when the graph being relabeled is a geometric graph, there exist vertex relabelings which greatly reduce the bandwidth of the measurement information matrix, thereby reducing the computational complexity of WSN localization; an efficient vertex relabeling algorithm which achieves this bandwidth-reduction is proposed as well. 
Furthermore, we derive an upper bound on the bandwidth of typical WSN configurations, which are modeled as random geometric graphs \cite{2010randomWSN}.
Through a numerical example of distance-based WSN localization, it is shown that the vertex relabeling step makes LB-EKF efficient and scalable, while retaining the performance characteristics of the EKF algorithm.

The rest of the paper is organized as follows; Section \ref{sec:problem} introduces the mathematical description of the WSN. Section \ref{sec:ekf} develops the EKF-based solution to the WSN localization problem. In Sections \ref{sec:matrix_bandwidth} and \ref{sec:graph_bandwidth}, the concepts of matrix and graph bandwidths are introduced, and it is shown how they can be exploited to make the EKF algorithm scalable. Sections \ref{sec:vr_gg} and \ref{sec:RandomGraphs} establish bounds on the performance of the proposed vertex relabeling approach on geometric graphs and random geometric graphs, respectively. Finally, Section \ref{sec:simulations} uses numerical simulations to validate the analysis in the preceding sections.

\textit{Notation:} Throughout the paper, $|\cdot|$ denotes the absolute value of a number as well as the cardinality of a set. The Euclidean norm is denoted by $\|\cdot\|$. $I_n$ is the $n\times n$ identity matrix and $\bold 0$ denotes a matrix of zeros (of appropriate dimensions). The symbols `$\times$' and `$\otimes$' denote the Cartesian and Kronecker products, respectively. $\textit{Po}(\lambda)$ denotes the Poisson distribution with the rate parameter $\lambda$.

\section{Problem Setup}
\label{sec:problem}

\subsection{Topology of the WSN}
Consider an undirected graph $\mathcal G \coloneqq (\mathcal V, \mathcal{E})$, where the set of vertices denoted by $\mathcal{V}$ represents the set of agents to be localized, and the set of edges $\mathcal E \subseteq \mathcal V \times \mathcal V$ represents the availability of pairwise measurements between them at timestep $k$, such as Euclidean distance measurements or relative bearing measurements. It is assumed that the vertices are ordered in an arbitrary way, such that $\mathcal V = \{1, 2, \dots , |\mathcal V|\}$. We call this a \textit{labeling} of the vertices.

Let the set of neighboring vertices of the $i^{th}$ vertex be $\mathcal N_i$. The network topology can be represented in matrix form using the \textit{Laplacian} of the graph, which is the matrix $\mathcal L = [\mathcal L_{ij}]$ defined as
\begin{equation}
\mathcal  L_{ij} = 
\left \{
\begin{aligned}
\sum_{k\in \mathcal N_i} 1&, \quad &  i=j\\
-1&,  & (i,j)\in \mathcal E \\
0&, & \textrm{otherwise}
\label{eq:laplacian}
\end{aligned}
\right.
\end{equation}

A \textit{realization} of $\mathcal G$ in $\mathbb R^d$ refers to an assignment of a vector in $\mathbb R^d$, called the position, to each of its vertices. The set of positions is denoted as
\begin{equation}
    \mathcal X \coloneqq \left\lbrace x_i(k)\ \big\vert\ x_i(k) \in \mathbb R^d,\ i\in \mathcal V \right\rbrace
\end{equation}
where the $i^{th}$ vertex is assigned the position $x_i(k)$. It is assumed that $\mathcal X$ is ordered similarly to $\mathcal V$, such that the WSN can be unambiguously represented by the realization $(\mathcal G, \mathcal X)$. For notational convenience, the timestep $k$ is omitted when referring to the edge set $\mathcal E$ and the position set $\mathcal X$, which may be time-varying due to the motion of the agents. 

A \textit{geometric graph} (in dimension $d$) is a graph which has a realization satisfying the following condition:
\begin{equation}
\|x_i(k)-x_j(k)\|\leq r \iff (i,j) \in \mathcal E
\label{eq:geom}
\end{equation}
for some $r \in (0, \infty)$.
\begin{assumption}[Limited Sensing Radius]
There exists some $r\in (0, \infty)$ such that the tuple $(\mathcal G, \mathcal X, r)$ satisfies (\ref{eq:geom}), i.e., a pairwise measurement is available between any two agents if and only if they are within some distance $r$ of each other. Consequently, $\mathcal G$ is a geometric graph.
\label{ass:geom}
\end{assumption}
Given Assumption \ref{ass:geom}, we refer to the corresponding distance $r$ as the \textit{sensing radius} of the WSN. Thus, geometric graphs can model the aspect of limited sensing/communication range in WSN applications; in the presence of obstructions in the environment, $\mathcal G$ may also be modeled as the subgraph of a geometric graph.

\subsection{Motion and Measurement Models}
\noindent
Define $x(k) \coloneqq [x_1(k)^\top x_2(k)^\top \dots x_{|\mathcal V|}(k)^\top]^\top$ as the concatenated vector of positions of the agents. The motion model of each agent is modeled as a stochastic single-integrator system, for which the discretized update model at timestep $k$, at agent $i$, is
\begin{equation}
    x_i(k+1) = x_i (k) + v_i(k) + p_i(k)
    \label{eq:system}
\end{equation}
where $v_i(k)$ is a known input, $p_i(k)$ is sampled from a zero-mean additive white Gaussian noise process with covariance $\sigma _p I_d$ and $x(0)$ is given by the initial conditions. The vector $v_i(k)$ is denoted as such because it represents the agent's velocity integrated over the discretization period. The noise $p_i(k)$ can capture a variety of uncertainties in the knowledge of $v_i(k)$, as well as environmental disturbances such as wind gusts. 
The single-integrator system model is commonly used in this literature \cite{observability2015, zhao2019bearing}, as in many real world applications, an informative estimate of the velocity can be obtained using intrinsic sensors with high sampling rates, such as Inertial Measurement Units (IMUs), so that the estimation of the velocity may be decoupled from the estimation of the position in order to keep the notation and analysis concise.

In order to model nonlinear measurements between the agents, such as inter-agent distance, received signal strength (RSS) and/or relative bearing measurements, we introduce the function $\phi : \mathbb R^d \times \mathbb R^d \rightarrow \mathbb R^m$, which is assumed to be differentiable almost everywhere. For instance, in the case of distance measurements, we can define $\phi(x_i(k), x_j(k)) \coloneqq \|x_i(k) - x_j(k)\|$. 
In order to ensure that the collective state vector $x(k)$ is observable, we assume that a subset of the vertices $\mathcal B \subseteq \mathcal V$, referred to as the set of \textit{beacons}, can observe their own state (e.g., using GPS measurements). Without loss of generality, let these be the first $|\mathcal B|$ vertices in $\mathcal V$.

The measurement model of the agents can be collectively written as
\begin{align}
    y(k) &= h(x(k)) + \begin{bmatrix}q(k)\\ r(k)\end{bmatrix}
    \label{eq:measurement}
\end{align}
\begin{center}
\begin{tikzpicture}[>=stealth,thick,baseline]
    
    \matrix [matrix of math nodes,left delimiter={[},right delimiter={]}](A){ 
    x_1(k)\\ x_2(k)\\ \vdots\\ x_{|\mathcal B|}(k)\\ \vdots \\ \phi(x_i(k),x_j(k))\\ \vdots \\
   };
    
    \begin{scope}[
    execute at begin node = $,
    execute at end node = $,
    ]
    \node[
     fit=(A-4-1)(A-4-1),
     inner xsep=28pt,inner ysep=0,
     label=left: h(x(k)) \coloneqq
    ](L) {};
    
  \node[
     fit=(A-6-1)(A-6-1),
     inner xsep=22pt,inner ysep=0,
     label={right, text=darkgray}: {(i,j) \in \mathcal E}
    ](L) {};
    \end{scope}
    

    \draw[->, darkgray](L.east)-- ([xshift=12pt]A-6-1.east);
    
    \draw[shorten >=2pt, shorten <=2pt, darkgray] ([xshift=81.5pt, yshift=-1pt]A-1-1.north east) --++(0:3mm)|-([xshift=78.5pt, yshift=3pt]A-4-1.south east) node[pos=.25, rotate=270, above]{\small $d|\mathcal B|$ entries};
    
    \draw[shorten >=2pt, shorten <=2pt, darkgray] ([xshift=92pt, yshift=-2pt]A-5-1.north east) --++(0:3mm)|-([xshift=92pt, yshift=0pt]A-7-1.south east) node[pos=.25, rotate=270, above]{\small $m|\mathcal E|$ entries};

\end{tikzpicture}
\end{center}
where $q(k)\in \mathbb R^{d|\mathcal B|}$ and $r(k)\in \mathbb R^{m|\mathcal E|}$ are the measurement noise vectors, which are sampled from independent zero-mean white Gaussian noise processes having covariances $\sigma_q I_{d|\mathcal B|}$ and $\sigma_r I_{m|\mathcal E|}$, respectively. The existence of one or more beacons (sometimes called \textit{anchors}) is a standard assumption in literature. For example, a necessary condition for observability in the case of distance-based localization in 2 dimensions (i.e., $d=2$) is known to be $|\mathcal B|\geq 2$, as this `pins' the estimates of two of the agents, preventing any continuous rigid motions of the entire network \cite{observability2015}. More generally, we make the following assumption.

\begin{assumption}[Observability]
The system (\ref{eq:system}), (\ref{eq:measurement}) satisfies the observability rank condition given in \cite{ekf1995} on some dense set $\mathcal O \subseteq R^{d|\mathcal V|}$ containing $x(0)$ in its interior. Concisely, this states that for all points $x'\in \mathcal O$, the observability matrix constructed using the Jacobian of $h(\cdot)$ at $x'$ is full rank.
\label{ass:observability}
\end{assumption}
It should be noted that in certain applications, Assumption \ref{ass:observability} can be checked in a fast, scalable manner. For instance, in distance or bearing-based localization in 2-dimensions, Assumption \ref{ass:observability} simplifies to purely combinatoric properties of the graph (called as the \textit{infinitesimal rigidity} condition in the literature) which can be checked in polynomial time (with respect to $|\mathcal V|$) \cite{observability2015, zhao2019bearing}.

\section{Extended Kalman Filter (EKF)}
\label{sec:ekf}
The Extended Kalman Filter (EKF) is a widely-used state estimation algorithm which is based on an improvement upon the linearized Kalman Filter, and is able to accomodate nonlinearities in the state and/or measurement models. 
Moreover, the EKF achieves the Cram\'er-Rao lower bound of the localization problem if it is linearized about the true trajectories of the agents \cite{ekfCRLB}. Since the true trajectories are not available in practice, the EKF is instead linearized about the estimated positions. Thus, its performance is asymptotically close to the lower bound. 
Given Assumption \ref{ass:observability} and some mild regularity conditions, it can be shown that the EKF estimate (almost surely) converges exponentially fast to a neighborhood of this asymptotic value \cite{ekf1999}.
In the WSN localization problem, the asymptotic error is dictated by the linearization error of the measurement model (\ref{eq:H}), which in turn depends on the second-order term in the Taylor series expansion of $\phi(\cdot, \cdot)$.

In order to implement the centralized EKF, the collective measurement function $h(x(k))$ must be linearized. Let the Kalman filter estimate at timestep $k$ be denoted $\hat x (k)$, such that its $i^{th}$ subvector $\hat x_i(k)$ is the agent $i$'s position estimate. 
Let the Jacobian of $h(x(k))$ computed at the current estimate $\hat x(k)$ be denoted as $H_k$, such that
\begin{align}
    H_k &\coloneqq \frac{\delta h}{\delta x}\Big\rvert_{\hat x(k)}\\[5pt]
    &= \begin{bmatrix}
    I_{d|\mathcal B|} \qquad \qquad \qquad \qquad \bold 0 \vspace{3pt}
    \\
    \vspace{3pt}
 \begin{bmatrix}
    & & & \vdots &  &  \\
    \dots & 0 & \frac{\delta \phi}{\delta x_i}\big\rvert_{\hat x (k)} & 0 & \dots & \frac{\delta \phi}{\delta x_j}\big\rvert_{\hat x(k)} & \dots \\
     & &  & \vdots  & & 
    \end{bmatrix}
    \end{bmatrix}
\label{eq:H}
\end{align}
where the submatrix is a sparse (mostly zeros) matrix of dimensions $m|\mathcal E| \times d |\mathcal V|$ (which is a consequence of the fact that measurements are only available between agents that are within each other's sensing radius). The measurement information matrix $S_k$ is defined as $S_k = H_k^\top R_k^{-1}H_k$, where
\begin{equation}
    R_k = \begin{bmatrix} 
    \sigma_q I_{d|\mathcal B|} & \bold 0\\
    \bold 0 & \sigma _r I_{m|\mathcal E|}
    \end{bmatrix}
\end{equation}
is the covariance of the measurement noise.
Using (\ref{eq:H}), we see that $S$ is a block diagonal matrix of the form
\begin{align}
    S_k = \begin{bmatrix}
    S^{(1)}_{k} & \bold 0 \\
    \bold 0 & S^{(2)}_{k}
    \end{bmatrix}
    \label{eq:Iblock}
\end{align}
where $S^{(1)}_{k} =\tfrac{1}{\sigma _q} I_{d |\mathcal B|}$ is a diagonal matrix and $S^{(2)}_{k}$ is a $d|\mathcal V| \times d|\mathcal V|$ sparse symmetric matrix whose sparsity pattern is further explored in Sec. \ref{sec:efficient}. 

The initial estimate of the EKF $\hat x(0)$ must be chosen to be sufficiently close to $x(0)$ to ensure that the algorithm converges \cite{ekf1999}. 
The estimate of the estimation error covariance at timestep $k$ is denoted as $M_k$, which can be initialized as
\[
M_0\coloneqq \mathbb E \left[\Big(\hat x(0)-x(0)\Big)\Big(\hat x(0)-x(0)\Big)^\top \right]\]
$M_k$ is updated at each timestep $k\geq 1$, using
\begin{align}
P_{k} &= M_{k-1} + \sigma_p I_{d|\mathcal V|}
\label{eq:P_update}\\
M_{k} &= \big(P_{k}^{-1} + S_{k}\big)^{-1}
\label{eq:M_update}
\end{align}
which are the same as the estimation error covariance update equations of the Kalman filter, with the key distinction being that the measurement information matrix $S_k$ depends on $\hat x(k)$, so that $M_k$ is only an estimate of, and not the true estimation error covariance of the EKF.
The estimate of the agents' positions $\hat x(k)$ is updated as per the recursion
\begin{align}
    \hat x(k+1) = \hat x(k) + v(k) + M_k H_k^\top R_k^{-1} \Big(y(k) - h\big(\hat x(k)\big)\Big)
    \label{eq:x_update}
\end{align}
where $v(k) \coloneqq [v_1(k)^\top\ v_2(k)^\top\ \dots\ v_{|\mathcal V|}(k)^\top]^\top$ is the concatenated velocity vector of the agents.

\section{Exploiting Sparsity through Vertex Relabeling}
\label{sec:efficient}
We see that the centralized EKF requires the inversion of large matrices, in (\ref{eq:M_update}). Even when $M_0$ is block diagonal, $P_k$ and $M_k$ are not block diagonal in general. This is because the pairwise measurements $\phi(\cdot, \cdot)$ cause the agents' estimation errors to be coupled, leading to non-zero off-diagonal blocks in the covariance matrices. The complexity of inverting the matrices in (\ref{eq:P_update}) and (\ref{eq:M_update}) using general matrix inversion algorithms is $O($\small$(d|\mathcal V|)$\normalsize$^3)$ \cite{cormen2022introduction}. Thus, it is of interest to explore the sparsity patterns of the matrix sequences $\lbrace P_k \rbrace$ and $\lbrace M_k \rbrace$ to determine when and how they can be inverted more efficiently for large-scale implementations.

\subsection{Approximation of $\lbrace M_k \rbrace$ as \\Banded Matrices}
\label{sec:matrix_bandwidth}
\noindent
The bandwidth of a matrix is defined as follows \cite{bandwidthsSurvey2014}.
\begin{definition}[Bandwidth of a matrix]
The bandwidth of an $n\times n$ matrix $A$ is
\begin{equation}
\textit{Bandwidth}(A) = \max_{A_{ij}\neq 0} \left( |i - j|  \right)
\label{def:band}
\end{equation}
\end{definition}
\noindent
Note that $\textit{Bandwidth(A)}$ counts how far the non-zero entries of $A$ are from its main diagonal, as
\begin{equation}
    |i-j| > \textit{Bandwidth}(A) \Rightarrow A_{ij}=0
\end{equation}
When the matrix $A$ has a small bandwidth, we say that it is a \textit{banded matrix}.
The complexity of many matrix operations reduces considerably when restricted to symmetric banded matrices \cite{kavcic2000matrices, liao2016improved_band_matrix_computation}, which are also faster to store and access from memory \cite{rcm2017}.

To ensure that the matrices in $\lbrace P_k\rbrace$ and $\lbrace M_k\rbrace$ are banded matrices, we need the following assumption.
\begin{assumption}
The initial estimates of the agents' positions are uncorrelated, i.e., $M_0$ is a block diagonal matrix.
\label{ass:init_cov}
\end{assumption}

Given Assumption \ref{ass:init_cov}, $M_0$ is a banded matrix with the bandwidth $d-1$. To proceed by induction, suppose that $M_{k-1}$ is a banded matrix.
From (\ref{eq:P_update}), it is observed that $P_k$ has the same sparsity pattern as $M_{k-1}$, and thus,
$P_k^{-1}$ can be approximated as a banded matrix. If $S_k$ also has a small bandwidth, then the covariance update step in (\ref{eq:M_update}) is the inversion of a banded matrix as well, ensuring that every matrix in the sequence $\lbrace M_k \rbrace$ has a small bandwidth. 

Given some number $L\in \mathbb N$, consider the EKF algorithm which uses the approximated error covariance matrix sequence $\lbrace \breve M^{-1}_k \rbrace$ in its computation, where
\[
\textit{Bandwidth}(\breve M^{-1}_k) = L
\]
Here, $\breve M^{-1}_k$ is chosen as the matrix of bandwidth $L$ which minimizes the information loss between $M^{-1}_k$ and itself. Such a matrix is unique, and can be computed using the L-banded matrix inversion algorithm introduced in \cite{kavcic2000matrices} and \cite{moura2008distributing}, summarized in Algorithm \ref{alg:LBMI}, where the notation $A^{m}_{n}$ denotes the $(n-m) \times (n-m)$ principle submatrix of the matrix $A$ spanning from row $m$ to row $n$, and from column $m$ to column $n$. The computation of $\breve M^{-1}_k$ requires the inversion of $L \times L$ matrices, which is a significant reduction in complexity from general matrix inversion when $L \ll |\mathcal V|$ \cite[Section 3]{kavcic2000matrices}.
We call this approach the \textit{$L$-Banded Extended Kalman Filter (LB-EKF)}, summarized in Algorithm \ref{alg:lbekf}, which is a family of estimation algorithms parameterized by $L$. 
\begin{algorithm}
\caption{L-Banded Matrix Inversion}
\begin{algorithmic}[1]
\REQUIRE $A \in \mathbb R^{n \times n}$, $L \in \lbrace 0, 1, \dots, n-1\rbrace$
\STATE Initialize a matrix of all zeros, as $Z \leftarrow \bold 0$
\FOR{$l = 1$ to $n-L$}
\STATE $Z^{l}_{l+L} \leftarrow Z^{l}_{l+L} + (A^{l}_{l+L})^{-1}$
\ENDFOR
\FOR{$l=2$ to $n-L$}
\STATE $Z^l_{l+L-1} \leftarrow Z^l_{l+L-1} - (A^l_{l+L-1})^{-1}$
\ENDFOR
\RETURN $Z$, which is the L-banded approximation of $A^{-1}$
\end{algorithmic}
\label{alg:LBMI}
\end{algorithm}
\begin{algorithm}
\caption{L-Banded Extended Kalman Filter (LB-EKF)}
\begin{algorithmic}[1]
\REQUIRE $\hat x(0)$, $M_0$, and  $L \in \lbrace 0, 1, \dots, n-1\rbrace$\\
\hspace{-15pt} At timestep $k$,
\STATE Concatenate the measurements in a vector $y(k)$
\STATE Compute the linearized measurement matrix $H_k$ using (\ref{eq:H})
\STATE Compute $P_k$ as per (\ref{eq:P_update}) and approximate $P_k^{-1}$ using the L-banded inverse of $P_k$ (as per Alg. \ref{alg:LBMI})
\STATE Approximate $M_{k}$ by the L-banded inverse of $(P_k^{-1}+S_k)$, as per (\ref{eq:M_update})
\STATE Compute the state estimate using (\ref{eq:x_update})
\end{algorithmic}
\label{alg:lbekf}
\end{algorithm}

When $L=0$, the approximated covariances $\lbrace \breve M_k \rbrace$ are diagonal matrices; when $L=1$, they are tridiagonal matrices, and so on. When $L=d|\mathcal V| - 1$, $\lbrace \breve M_k \rbrace = \lbrace M_k \rbrace$, i.e., the LB-EKF coincides with the EKF.
The authors of \cite{moura2008distributing} remark that the sequence $\lbrace \breve M^{-1}_k \rbrace$ may diverge when $L$ is chosen to be too small, so it should be chosen to be large enough, with its exact value depending on the model matrices as well as the labeling of the WSN.

\subsection{Minimal Bandwidth of $\mathcal G$}
\label{sec:graph_bandwidth}
In Section \ref{sec:matrix_bandwidth}, the LB-EKF algorithm was motivated by assuming that the matrix $S_k$ has a small bandwidth. Thus, the main challenge in making the EKF recursion efficient is in ensuring that $S_k$ is a banded matrix. The bandwidth of $S_{k}$ can be reduced by exploiting the properties of $\mathcal G$, the underlying geometric graph of the WSN. To this end, the following definition of the \textit{minimal bandwidth} of a graph is introduced.
\begin{definition}[Minimal bandwidth of a graph]
Let $\Pi$ be the set of bijections from $\mathcal V$ to $\mathcal V$, such that
\begin{equation}
    \Pi = \left \lbrace \pi\ \vert\ \pi:\mathcal V \rightarrow \mathcal V , \textit{ $\pi(\cdot)$ is a bijection} \right\rbrace
\end{equation}
The minimal bandwidth of a graph $\mathcal G=(\mathcal V , \mathcal E))$ is defined as
\begin{equation}
    \varphi _{\textrm{min}}(\mathcal G) = \min_{\pi \in \Pi}\left( \max_{(i,j)\in \mathcal E}(|\pi(i)-\pi(j)|) \right)
\end{equation}
\end{definition}
\noindent
In the foregoing definition, the bijection $\pi$ is a \textit{relabeling} of the graph's vertices. Corresponding to each relabeling $\pi$, there is a unique permutation matrix $P_\pi$, such that the Laplacian matrix of $\mathcal G$ becomes $P_\pi \mathcal L P_\pi ^\top$ after the relabeling. Thus, we see that the notions of bandwidth for matrices and graphs are related as follows.
\begin{equation}
    \varphi_{\textrm{min}} (\mathcal G) = \min_{P_\Pi}\left(\textit{Bandwidth}(P_\pi \mathcal L P_\pi^\top)\right)
\end{equation}
where $P_\Pi$ is the set of all $|\mathcal V|\times |\mathcal V|$ permutation matrices, and the minimum is taken over the bandwidths of all matrices that are permutation-similar to $\mathcal L$. In the literature, $\varphi _{\textrm{min}} (
\mathcal G )$ is simply referred to as the bandwidth of the graph $\mathcal G$ \cite{embedAndProject}. However, we call it the minimal bandwidth in order to avoid any confusion between the two definitions of bandwidth. The bandwidth of a graph refers to the bandwidth of its Laplacian matrix, which may or may not be minimal (i.e., equal to $\varphi _{\textrm{min}}(\mathcal G)$) depending on how the graph's vertices are labeled.

Our interest in connecting these two definitions of bandwidth comes from the following proposition.
\begin{lemma}
The matrix $S^{(2)}_{k}$ in (\ref{eq:Iblock}) has the same sparsity pattern as $\mathcal L\otimes \mathbbm{1}_d$, where $\mathbbm 1_d$ is the $d \times d$ matrix where each entry is $1$.
\label{prop:sparsity}
\end{lemma}
\begin{proof}
Partition $S^{(2)}_{k}$ into $d \times d$ blocks, and consider the $ij^{th}$ block. When $i=j$ the block is non-zero in general. For $i\neq j$, the $ij^{th}$ block corresponds to the vertex pair $(i,j)\in\mathcal V \times \mathcal V $, and has the form 
\begin{equation*}
\frac{1}{\sigma_r}\sum_{l}\left(\frac{\delta h_l}{\delta x_i}\right)^\top\frac{\delta h_l}{\delta x_j}
\end{equation*}
which is only non-zero if there is an entry in $h(x)$ that depends on both $x_i$ and $x_j$. This happens precisely when $(i,j)\in \mathcal E$. Using the definition of $\mathcal L$ in (\ref{eq:laplacian}), we see that the blocks of $S^{(2)}_{k}$ have the sparsity pattern of $\mathcal L$.
\end{proof}
\noindent
Note that the bandwidth of $S_k$ is dictated by the bandwidth of $S^{(2)}_k$, as $S^{(1)}_k$ is a diagonal matrix. Using Proposition \ref{prop:sparsity} and a simple counting argument, we have
\begin{equation}
\textit{Bandwidth}(S_{k}) = d\big(\textit{Bandwidth}(\mathcal L)+1\big)-1
\end{equation}
Thus, the complexity of implementing LB-EKF scales with the bandwidth of the Laplacian matrix $\mathcal L$, which in turn depends on how the vertices of the graph $\mathcal G$ are labeled. Unfortunately, for any given graph, there exists a labeling such that the Laplacian matrix has maximal bandwidth, which is $|\mathcal V|-1$ (corresponding to the case where agent $1$ is connected to agent $n-1$), so the bandwidth of $S_{k}$ can be arbitrarily large. This makes the problem of finding a good relabeling of $\mathcal V$ an important sub-problem for large-scale WSN localization.
\subsection{Vertex Relabeling}
\label{sec:vr_gg}
Naturally, the most desirable labeling of the WSN agents is the one that corresponds to the minimal bandwidth of the corresponding graph. Much effort has gone into finding efficient algorithms for finding the minimal bandwidth of a graph, which is known to be an NP-complete problem even when restricted to geometric graphs \cite[Theorem 3.1]{bandwidth_NPcompleteness2001Penrose}. 
Instead, heuristic approaches are used in practice, such as the Reverse Cuthill-McKee \cite{rcm2017} and Embed and Project \cite{embedAndProject} algorithms, amongst others \cite{bandwidthsSurvey2014}. Since these algorithms do not exploit the properties of geometric graphs, we propose a novel low-complexity algorithm for reducing the bandwidth of geometric graphs, called the vertex relabeling algorithm, and obtain an upper bound on the resulting bandwidth. 

To arrive at the algorithm, note that a relabeling of the vertices of a geometric graph corresponds to a map from points in $\mathbb R^d$ onto integers on the number line. This suggests that one can choose a suitable line in $\mathbb R^d$ and project the positions of the agents onto this line. In the absence of any additional heuristics about the sensor network configuration, we can simply choose the $\textrm{x}$-axis 
as this line. Combining these ideas gives us the vertex relabeling algorithm (Alg. \ref{alg:vr}). As the algorithm only involves a sorting operation, its complexity is $O\big(|\mathcal V|\log(|\mathcal V|)\big)$ \cite{cormen2022introduction}.

\begin{algorithm}
\caption{Vertex Relabeling (VR) Algorithm for Geometric Graphs}
\begin{algorithmic}[1]
\REQUIRE A realization $(\mathcal G,\mathcal X)$ of a geometric graph.
\STATE Sort the vertices in $\mathcal V$ in increasing order of the $\textrm{x}$-coordinates of $\tilde x_i$. Resolve any ties by sorting as per their $\textrm{y}$-coordinate, and so on.
\RETURN The reordered vertex set $\mathcal V_\pi$.
\end{algorithmic}
\label{alg:vr}
\end{algorithm}
\begin{figure}[h]
\centering
        \begin{subfigure}{0.2\textwidth}
            \centering
            \includegraphics[width=1.0\textwidth]{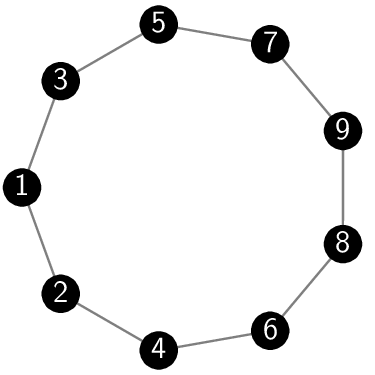}
        \end{subfigure}
        \begin{subfigure}[t]{0.28\textwidth}
            \centering
            \includegraphics[width=1.0\textwidth]{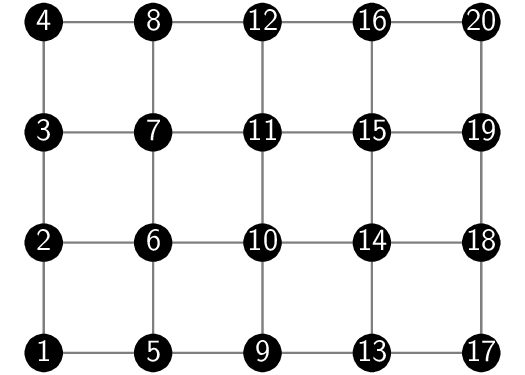}
        \end{subfigure}
        \caption{Labeling a cycle graph and a grid graph using Algorithm \ref{alg:vr} achieves the minimal bandwidth (2 and 4, respectively) on either graph.}
        \label{fig:graphs_2}
\end{figure}
It remains to be justified that Algorithm \ref{alg:vr} is indeed effective in reducing the bandwidth of a general geometric graph.
Let the bandwidth achieved by the vertex relabeling algorithm (Alg. \ref{alg:vr}) on a geometric graph $\mathcal G$ be denoted as $\varphi (\mathcal G, \mathcal X)$. Here, $\mathcal X$ is the set of positions of the agents.
It can be noted that $\varphi(\mathcal G, \mathcal X) = \varphi_{\textrm{min}}(\mathcal G)$ for the graphs in Fig. \ref{fig:graphs_2}. More generally, a useful bound on the performance of Algorithm \ref{alg:vr} on geometric graphs can be derived as follows.
\begin{theorem}[Upper bound of $\varphi(\mathcal G, \mathcal X)$]
The bandwidth achieved by Algorithm \ref{alg:vr} on a geometric graph of sensing radius $r$, $\mathcal G=(\mathcal V, \mathcal E)$, has the following upper bound:
\begin{align}
\varphi (\mathcal G, \mathcal X) &\leq
\max_{a\in \mathbb R}\Big
( | \mathcal{X}\cap \textit{R}_\emph{x}(a, a+r) | \Big)  \nonumber
\\
&\eqqcolon \varphi _{\textrm{max}}(\mathcal X, r)
\label{eqn:vr_bound}
\end{align}
where $\textit{R}_\emph{x}(a_1,a_2)\coloneqq [a_1, a_2]\times\mathbb R^{d-1}$.
\label{theorem:vr_bound}
\end{theorem}
\begin{proof}
Recall that $\varphi(\mathcal G, \mathcal X)$ is the bandwidth of the Laplacian matrix of the graph after vertex relabeling. When $\mathcal G$ is the subgraph of another graph on the same vertex set, $\mathcal{G}' = (\mathcal{V}, \mathcal{E}')$, we have
\begin{equation}
    \mathcal{E} \subseteq \mathcal{E}' \Rightarrow \varphi(\mathcal G, \mathcal X) \leq \varphi(\mathcal G', \mathcal X)
    \label{eq:subgraph_inequality}
\end{equation}
since adding edges to a graph corresponds to changing the corresponding entries of its Laplacian matrix from 0 to 1, which can only increase the bandwidth.
To prove the main result, consider the graph $\mathcal{G}'$ whose edge set is as follows:
\begin{equation}
    x_j \in \textit{R}_\textrm{x}(x_i, x_i+r) \Rightarrow (i,j) \in \mathcal E'
\end{equation}
$\forall i \in \mathcal V$, where $\textit{R}_\textrm{x}(a_1,a_2)$ is as defined in the statement of the theorem, and can be thought of as an infinitely long strip in Euclidean space. Since
\begin{align}
(i,j) \in \mathcal E  & \Rightarrow  \|x_i - x_j\| \leq r \nonumber \\ & \Rightarrow   x_j \in \textit{R}_\textit{x}(x_i, x_i+r)\ \textrm{or}\  x_i \in \textit{R}_\textrm{x}(x_j, x_j+r) & \nonumber \\
&\Rightarrow (i,j) \in \mathcal E'
\end{align}
we have that $\mathcal{E}\subseteq \mathcal{E}'$, i.e., $\mathcal{G}$ is indeed a subgraph of $\mathcal{G}'$.
Using the definitions of matrix bandwidth and graph Laplacian, we have
\begin{align}
\varphi(\mathcal G', \mathcal X) &= \max_{(i,j)\in \mathcal E'} \left( |i - j|  \right)\\
&= \max_{i \in \mathcal V}\Big(\max \left(
\left\lbrace j - i \ \big|\ j\geq i, (i,j)\in \mathcal E' \right\rbrace
\right)
\Big)
\\
&= \max_{i\in \mathcal V}\Big
( | \mathcal{X}\cap \textit{R}_\textrm{x}(x_i, x_i+r) | - 1 \Big) 
\end{align}
where the last equality follows from the construction of $\mathcal G'$. Using (\ref{eq:subgraph_inequality}), we arrive at the desired inequality (\ref{eqn:vr_bound}).
\end{proof}
Clearly, $\varphi _{\textrm{max}}(\mathcal X, r)$ bounds the minimal bandwidth $\varphi _{\textrm{min}} (\mathcal G)$ of a geometric graph as well. The following corollary shows that this bound can be further improved by minimizing (\ref{eqn:vr_bound}) over the set of realizations that induce the same edge set $\mathcal E$.
\begin{corollary}[Upper bound of $\varphi_{\textrm{min}}(\mathcal G)$]
\label{corollary:sig_min}
For a geometric graph $\mathcal G=(\mathcal V, \mathcal E)$,
\begin{equation}
\varphi _{\textrm{min}} (\mathcal G) \leq \min_{(\mathcal X', r') \in \mathcal{R}_\mathcal{G}} 
\varphi _{\textrm{max}}(\mathcal X', r')
 \label{corollary_bound}
\end{equation}
where $\mathcal{R}_\mathcal{G}$ is the set of tuples defined as
\begin{equation}
 \mathcal{R}_\mathcal{G} \coloneqq  \left\lbrace 
 (\mathcal X', r')\ \big\vert \ (\mathcal G, \mathcal X', r') \textit{ is a geometric graph} 
 \right\rbrace
 \label{eqn:R_g_defn}
\end{equation}
\end{corollary}
\begin{proof}
Observe that $\varphi _{\textrm{min}}(\mathcal G)$ is completely determined by the edge set $\mathcal E$ of the graph. Thus, the inequalities
\begin{equation}
\varphi_{\textrm{min}}(\mathcal G)  \leq  \varphi  (\mathcal G, \mathcal X')\leq \varphi _{\textrm{max}}(\mathcal X', r')
\label{eq:bound_relationship}
\end{equation}
hold for all $(\mathcal X', r') \in \mathcal R_{\mathcal G}$, giving us (\ref{corollary_bound}).
\end{proof}
Corollary \ref{corollary:sig_min} is presented here only as a tangential result, because the minimization in (\ref{corollary_bound}) may be intractable, whereas the bound given in Theorem \ref{theorem:vr_bound} is achievable in $O(|\mathcal V||\log |\mathcal V|)$ time.
Equation (\ref{eq:bound_relationship}) summarizes the relationship between the various quantities introduced in this section.

It remains to be established whether the bound achieved by the vertex relabeling algorithm (Alg. \ref{alg:vr}) given in Theorem \ref{theorem:vr_bound} is indeed small. A tight bound on the minimal bandwidth of graphs is 
\begin{equation}
\varphi _{\textrm{min}} (\mathcal G) \leq |\mathcal V| - \max_{i,j \in \mathcal V}\left(\textit{d}(i,j)\right)
\label{eqn:general_upper_bound}
\end{equation}
where $d(\cdot,\cdot)$ is the length of the shortest path between a given pair of vertices in a graph \cite[Th. 3.2.1]{chinn1982bandwidth}. This bound is not always useful for geometric graphs; in fact, one can observe that the bound (\ref{eqn:vr_bound}) is more informative than (\ref{eqn:general_upper_bound}) on the cycle graph in Fig. \ref{fig:graphs_2}. Moreover, computing the second term of (\ref{eqn:general_upper_bound}) itself requires up to $O(|\mathcal V|^3)$ operations \cite[Section 25.2]{cormen2022introduction}, whereas Algorithm \ref{alg:vr} can be viewed as an efficient algorithm for bounding the bandwidth of a geometric graph. Finally, it is shown in the next section that the vertex relabeling algorithm proves to be highly effective at reducing the graph bandwidth in typical WSN scenarios.

\subsection{Random Geometric Graphs}
\label{sec:RandomGraphs}
Usually, WSNs are deployed to accomplish some underlying objective, such as in surveillance networks, ad-hoc mesh networks and search-and-rescue missions. In such applications, it is typical for the agents to be deployed with some expected spatial density, whereas the exact number of agents in a given region may not be known a priori. For simplicity of presentation, we assume that the expected spatial density of agents is uniform across the region of application, though the forthcoming results could be readily generalized to the case when the density is non-uniform. A commonly used model for large-scale WSNs which encapsulates these ideas is that of \textit{random geometric graphs} \cite{2010randomWSN,2010PoissonRGG, dimakis2006geographic}, which are characterized as follows.

\begin{algorithm}
\caption{Generative mechanism for random geometric graphs}
\begin{algorithmic}[1]
\REQUIRE The domain $D\subseteq \mathbb R^d$,  the rate parameter $\lambda$, and the sensing radius $r\in(0, \infty)$ 
\STATE Generate a set of points $\mathcal X^{(\lambda)} = \lbrace x^{(\lambda)}_1, x^{(\lambda)}_2, \dots \rbrace$ in the domain $D$ using the homogeneous Poisson point process with rate parameter $\lambda$.
\STATE Construct the vertex set $\mathcal V^{(\lambda)}$, with a vertex assigned to each of the points.
\STATE Construct the edge set $\mathcal E^{(\lambda)}$ as per the constraint on geometric graphs (\ref{eq:geom}), with the sensing radius $r$.
\RETURN $\mathcal G^{(\lambda)}\coloneqq (\mathcal V^{(\lambda)}$, $\mathcal E^{(\lambda)})$ and $\mathcal X^{(\lambda)}$
\end{algorithmic}
\label{alg:PRGG}
\end{algorithm}
\begin{definition}[Random Geometric Graph \cite{gilbert1961random}]
Any realization $(\mathcal G^{(\lambda)}, \mathcal X^{(\lambda)})$ generated using Algorithm \ref{alg:PRGG} is called a random geometric graph.
\end{definition}
The choice of the Poisson point process as the underlying mechanism for generating the locations of WSN agents is well-motivated in the literature as well \cite{haenggi2012stochastic, WSN2021}. 
The rate parameter $\lambda$ of the process dictates the spatial density of agents. Lemma \ref{lem:poisson} lists some of the important properties of random geometric graphs which follow from its definition.

\begin{lemma}[Properties of random geometric graphs \cite{daley2008poissonTheory}]
A random geometric graph $(\mathcal G^{(\lambda)}, \mathcal X^{(\lambda)})$ in $D\subseteq \mathbb R^d$ has the following properties: 
\begin{enumerate}
\item For any region $ A \subseteq D$, the distribution of the points $\mathcal X^{(\lambda)}\cap A$ conditioned on the number of points $|\mathcal X^{(\lambda)}\cap A|$ is the multivariate uniform random distribution over the support $A$.
\item For any region $ A \subseteq D$,
\begin{equation}
    |\mathcal X^{(\lambda)} \cap A|\sim \textit{Po}\left(\lambda\cdot\textit{Vol}(A)\right)
    \label{eqn:subsetsPoisson}
\end{equation}
where $\textit{Vol}(\cdot)$ is the $d$-dimensional volume. 
\item The expected number of vertices in $\mathcal G^{(\lambda)}$ is
\begin{equation}
\mathbb E\big[|\mathcal V^{(\lambda)}|\big]  =  \mathbb E\big[|\mathcal X^{(\lambda)} \cap D|\big] = \lambda\cdot\textit{Vol}(D)
\label{eqn:totalPoisson}
\end{equation}
\end{enumerate}
\label{lem:poisson}
\end{lemma}
\noindent
Consider a random geometric graph in $D=[0,l]^d$. 
Using (\ref{eqn:subsetsPoisson}), and the fact that $|\mathcal X \cap D^c|=0$, we can see that
\begin{equation}
    | \mathcal{X}^{(\lambda)}\cap \textit{R}_\textrm{x}(a_1,a_2) |
    \sim \textit{Po}(\lambda l(a_2-a_1))
    \label{1dPoisson}
\end{equation}
where $R_\textrm{x}(a_1,a_2)$ is defined as in Theorem \ref{theorem:vr_bound}. Equivalently, $| \mathcal{X}^{(\lambda)}\cap \textit{R}_\textrm{x}(0,a) |$ is a 1-dimensional Poisson point process indexed by $a\in [0,l]$, having the rate parameter $\lambda l$. To see this, one can apply the property (\ref{eqn:subsetsPoisson}) to this 1-dimensional process to yield (\ref{1dPoisson}).

Using Theorem \ref{theorem:vr_bound}, we have
\begin{align}
\varphi (\mathcal G^{(\lambda)}, \mathcal X^{(\lambda)}) & \leq \varphi_\textit{max} (\mathcal X^{(\lambda)}, r) \nonumber \\
&= \max_{0\leq a\leq l-r}\Big( | \mathcal{X}^{(\lambda)}\cap \textit{R}_\textrm{x}(a,a+r) | \Big)
\label{eq:poisson_1d_bound}
\end{align}
The last term in (\ref{eq:poisson_1d_bound}) is a random variable, called the \textit{scan statistic} of the 1-dimensional Poisson point process (\ref{1dPoisson}). As tractable expressions (e.g., elementary functions of $\lambda$ and $l$) of moments, conditional moments or tail probabilities of the scan statistic are unknown \cite{glaz2001scan}, it is difficult to characterize an upper bound for the scan statistic. It is easier to study the term inside the $\max(\cdot)$ operation instead. Using the property (\ref{eqn:subsetsPoisson}) of the 1-dimensional process (\ref{1dPoisson}), one can observe that
\begin{equation}
\mathbb E \left[| \mathcal{X}^{(\lambda)}\cap \textit{R}_\textrm{x}(a,a+r) | \right] = \lambda lr 
\end{equation}
whereas from (\ref{eqn:totalPoisson}), we have $E\left[|\mathcal V^{(\lambda)}| \right] = \lambda l^2$, so that
\begin{align}
\mathbb E \left[| \mathcal{X}^{(\lambda)}\cap \textit{R}_\textrm{x}(a,a+r) | \right] = r \sqrt{\lambda \mathbb E\left[|\mathcal V^{(\lambda)}| \right]}
\label{eq:sublinear_growth}
\end{align}
This indicates that the upper bound on $\varphi(\mathcal G^{(\lambda)}, \mathcal X^{(\lambda)})$ is sublinear in the number of agents, unlike the general upper bound $(\ref{eqn:general_upper_bound})$ which is linear. Indeed, the preceding argument is confirmed through simulations in Section \ref{sec:simulations}, specifically in Fig. \ref{fig:scan_statistic}.
The sublinear growth of the bandwidth achieved by the vertex relabeling algorithm, combined with its ease of implementation, motivates its use in large-scale WSN localization.
\section{Numerical Simulations}
\label{sec:simulations}
In this section, we first illustrate the relevance of matrix and graph bandwidth in large-scale WSN applications. Subsequently, the LB-EKF algorithm (Alg. \ref{alg:lbekf}) is simulated in conjunction with the vertex relabeling algorithm (Alg. \ref{alg:vr}) to demonstrate the efficacy of the proposed approach. Finally, it is verified that the bandwidth achieved by the vertex relabeling algorithm grows sublinearly with respect to the number of vertices on random geometric graphs.

\subsection{L-Banded Matrix Inversion}
To validate the use of L-banded inversion to approximate covariance matrices, we generated random $500 \times 500$ symmetric matrices of varying bandwidths; the non-diagonal elements were drawn uniform randomly from $(-1,1)$, the diagonal elements were set to $15$ and finally the matrix was normalized.
For each generated matrix, its inverse was approximated by a matrix of bandwidth $L$ using Algorithm \ref{alg:LBMI}, with $L$ ranging from $0$ to $50$. Figure \ref{band_inversion} shows the approximation error as a function of $L$, the bandwidth of the approximated matrix. For $L\gg 1$, the complexity of $L$-banded matrix inversion approaches that of general matrix inversion and the approximation error goes to 0. For smaller values of $L$, the approximation error decreases steeply until it is greater than the bandwidth of the matrix being inverted, illustrating that covariance matrices having a small bandwidth can be inverted efficiently with reasonable accuracy.
\begin{figure}[h]
\includegraphics[width=0.50\textwidth]{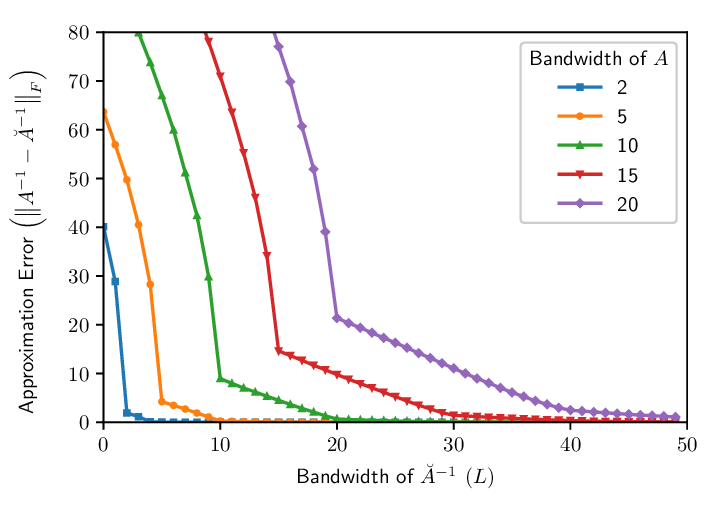}
\caption{The error (in Frobenius norm) in the approximation of the inverse of a $500\times 500$ covariance matrix $A^{-1}$ by another covariance matrix $\breve A^{-1}$ of a given bandwidth.}
\label{band_inversion}
\end{figure}

\subsection{LB-EKF with Vertex Relabeling} Next, we demonstrate the proposed WSN localization approach using the example of distance-based localization. Figure \ref{fig:network_init} depicts the initial configuration of $30$ agents in a $40m \times 40m$ domain. Of these, $8$ agents are beacons, i.e., they are able to measure their own position using GPS sensors with noise variance of $2 m^2$. All agents are able to measure their distance from their neighboring agents, with a variance of $10 m^2$. Thus, the agents follow the measurement model (\ref{eq:measurement}), with $\phi(x_i(k), x_j(k)) \coloneqq \|x_i(k) - x_j(k)\|$.
\begin{figure}[h]
\centering
\includegraphics[width=0.41\textwidth, trim={1.5cm 1.5cm 1.5cm 1.5cm}, clip]{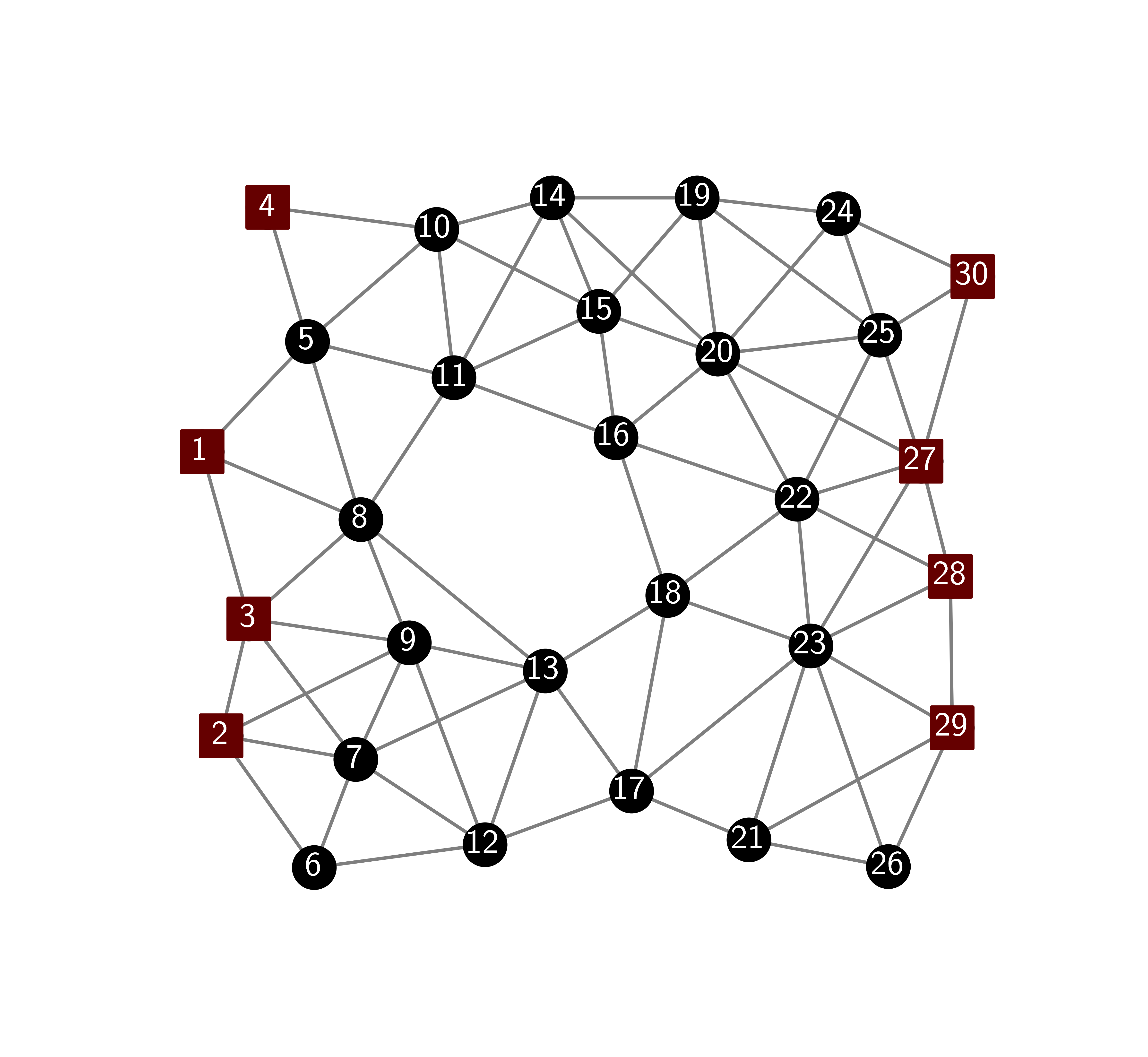}
\caption{Initial positions of the WSN; the red squares indicate the beacons (agents which can measure their own position).}
\label{fig:network_init}
\end{figure}
The sensing radius was chosen as $15m$, yielding the geometric graph shown in Fig. \ref{fig:network_init}. The vertex relabeling algorithm was used to reduce the bandwidth of the graph to $8$. The WSN was simulated for 100 timesteps using the LB-EKF algorithm with vertex relabeling (denoted as LB-EKF+VR), with an initial estimation error variance of $5m^2$ and process noise variance of $0.02 m^2$. Note that the dimension of the collective position vector of the WSN is $60$. We implemented LB-EKF+VR with $L=20$ as well as the original EKF algorithm (which is equivalent to LB-EKF with $L=59$), with the latter serving as a basis for comparison.
From (\ref{eq:M_update}), we see that EKF requires the inversion of $59 \times 59$ matrices, whereas from Algorithms \ref{alg:LBMI} and \ref{alg:lbekf} we know that LB-EKF only inverts $20\times 20$ matrices. For clarity of presentation, the labeling of the vertices is kept consistent between the plots.

\begin{figure}[h]
\centering
\includegraphics[width=0.45\textwidth]{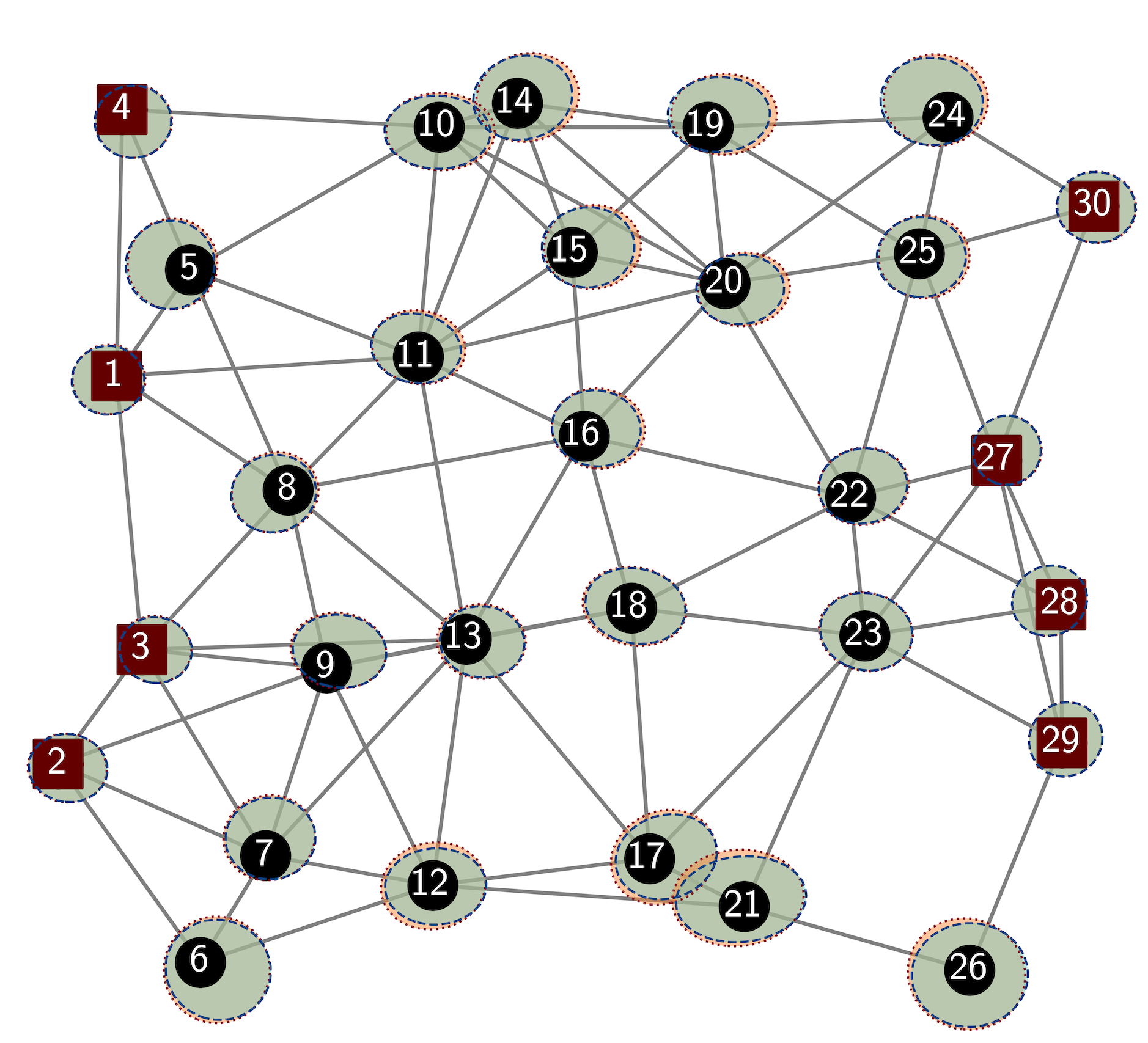}
\caption{Covariance ellipses of the estimates computed using LB-EKF+VR (blue shade) plotted over those of EKF (orange shade). The resulting green hue indicates that the ellipses of the two algorithms are overlapping.}
\label{fig:network_final}
\end{figure}
Fig. \ref{fig:network_final} depicts the \textit{covariance ellipses} of both algorithms at timestep 100, i.e., level sets of the form
\begin{equation*}
    \big(x - \hat x_i(k)\big)^\top \breve M_k^{(i)^{-1}}\big(x - \hat x_i(k)\big) \leq 20
\end{equation*}
where $\breve M_k^{(i)}$ is the $i^{th}$ diagonal block element of $\breve M_k$, and the number $20$ is chosen for visual emphasis. Note that the ellipses are centred around the estimated position vector of either algorithm at each agent.
The covariance ellipses allow us to visualize how well the banded error covariance $\breve M_k$ of the LB-EKF+VR algorithm approximates the true error covariance $M_k$ of the EKF algorithm. It can be seen that the ellipses are nearly coincident, as the parameter $L$ of LB-EKF+VR is large enough to capture the cross-covariances between agents. The estimates of LB-EKF+VR are slightly overconfident (have smaller ellipses) when compared to EKF, as LB-EKF effectively ignores the cross-covariances between agents that are far away from each other in terms of the labeling. 

The above experiment is repeated for $5000$ Monte Carlo trials. To demonstrate the importance of vertex relabeling, we implemented EKF, LB-EKF with vertex relabeling (LB-EKF+VR) as well as LB-EKF without vertex relabeling, i.e., using an arbitrary labeling of the vertices. Figure \ref{fig:mse_1} shows the mean squared error (MSE) of the individual agents, defined as
\begin{equation*}
    \textrm{MSE}_i(k) = \|x_i(k) - \hat x_i(k)\|^2
\end{equation*}
\begin{figure}[h]
\centering
\includegraphics[width=0.45\textwidth]{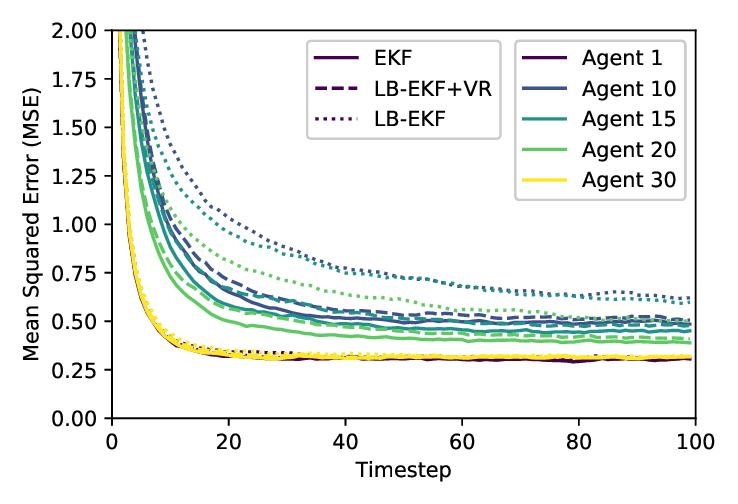}
\caption{MSE of the position estimates of the agents, computed using EKF, LB-EKF with vertex relabeling and LB-EKF without vertex relabeling, averaged over 5000 Monte Carlo trials.}
\label{fig:mse_1}
\end{figure}
\begin{figure}[h]
\centering
\includegraphics[width=0.42\textwidth]{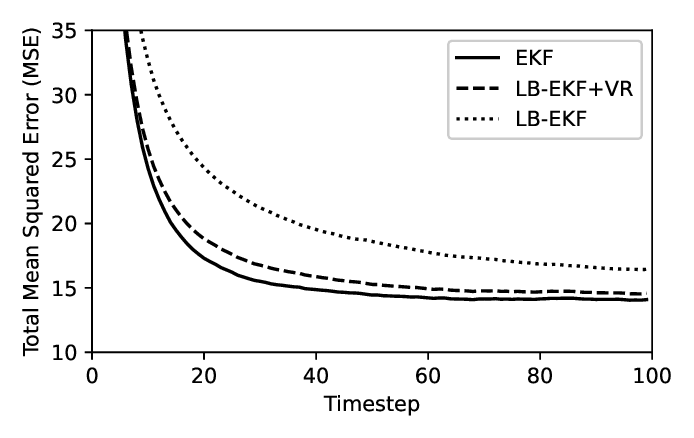}
\caption{Total MSE of all the agents in the WSN, computed using EKF, LB-EKF with vertex relabeling and LB-EKF without vertex relabeling, averaged over 5000 Monte Carlo trials.}
\label{fig:mse_2}
\end{figure}
Figure \ref{fig:mse_2} shows the total MSE of the network, $\sum _{i\in \mathcal V} \textrm{MSE}_i(k)$. While the three algorithms achieve comparable performance at the beacon agents (i.e., agents 1 and 30 in Fig. \ref{fig:mse_1}), the LB-EKF algorithm without vertex relabeling has worse performance at other agents. The MSE of LB-EKF+VR is close to that of EKF, in both the transient and steady-state segments of the simulation. This demonstrates that the vertex relabeling approach allows us to greatly reduce the computational burden of large-scale WSN localization while retaining reasonable estimation accuracy and convergence properties.

\subsection{Vertex Relabeling of Large-Scale WSNs}
Lastly, we use numerical simulations to validate that the upper bound of the bandwidth achieved by Algorithm \ref{alg:vr}, $\varphi _{\textit{max}}(\mathcal X^{(\lambda)}, r)$, grows sublinearly with respect to the number of vertices on random geometric graphs. To this end, we generate independent samples from the 1-dimensional uniform random distribution in the domain $[0,l]$, for varying values of $l$. From Property 1 of Lemma \ref{lem:poisson}, we know that the points generated in this manner correspond to the points generated by a 1-dimensional Poisson point process in $[0,l]$, conditioned on the total number of points.
Furthermore, from (\ref{1dPoisson}), we know that this is equivalent to generating points in a 3-dimensional (cubic) domain, since the vertex relabeling algorithm sorts the points based on their $\textrm{x}$-coordinate alone\footnote{The probability of generating two points with coincident $\textrm{x}$-coordinates is 0, so the algorithm does not encounter any ties.}.
Together, the foregoing observations allow us to compute the value of $\varphi _{\textit{max}}(\mathcal X^{(\lambda)}, r)$ for a large number of points in a $[0,l]^2$ domain efficiently.

The sensing radius $r$ was chosen as $15m$. 
The simulations were repeated for the rate parameters $\lambda =  0.005,\ 0.01,\ 0.02,\ 0.05 $ and $ 0.1$. To put these numbers into perspective, the expected number of agents in a $40m \times 40m$ domain is
$8$, $16$, $32$, $80$ and $160$, respectively, for the corresponding choices of $\lambda$. Thus, the chosen values of $\lambda$ capture both sparse and densely connected networks.

Figure \ref{fig:scan_statistic} shows the value of $\varphi _{\textit{max}}(\mathcal X^{(\lambda)}, r)$ plotted against the total number of agents in the graph. It can be noted that the bandwidth indeed grows very slowly compared to the maximal bandwidth of the graph, which grows linearly. The growth of the bandwidth slightly overshoots the predicted growth of (\ref{eq:sublinear_growth}), since the analysis leading to (\ref{eq:sublinear_growth}) ignored the $\textrm{max}(\cdot)$ operation involved in the computation of $\varphi _{\textit{max}}(\mathcal X^{(\lambda)}, r)$. The standard deviation of the achieved bound is quite small as well; this is because as the number of agents increases, the $\textrm{max}(\cdot)$ operation takes the maximum over a large number of agents, so that the law of large numbers comes into effect. Consequently, the vertex relabeling algorithm achieves a small bandwidth on typical WSN configurations with high probability.

\begin{figure*}[t]
\centering
\includegraphics[width=0.90\textwidth]{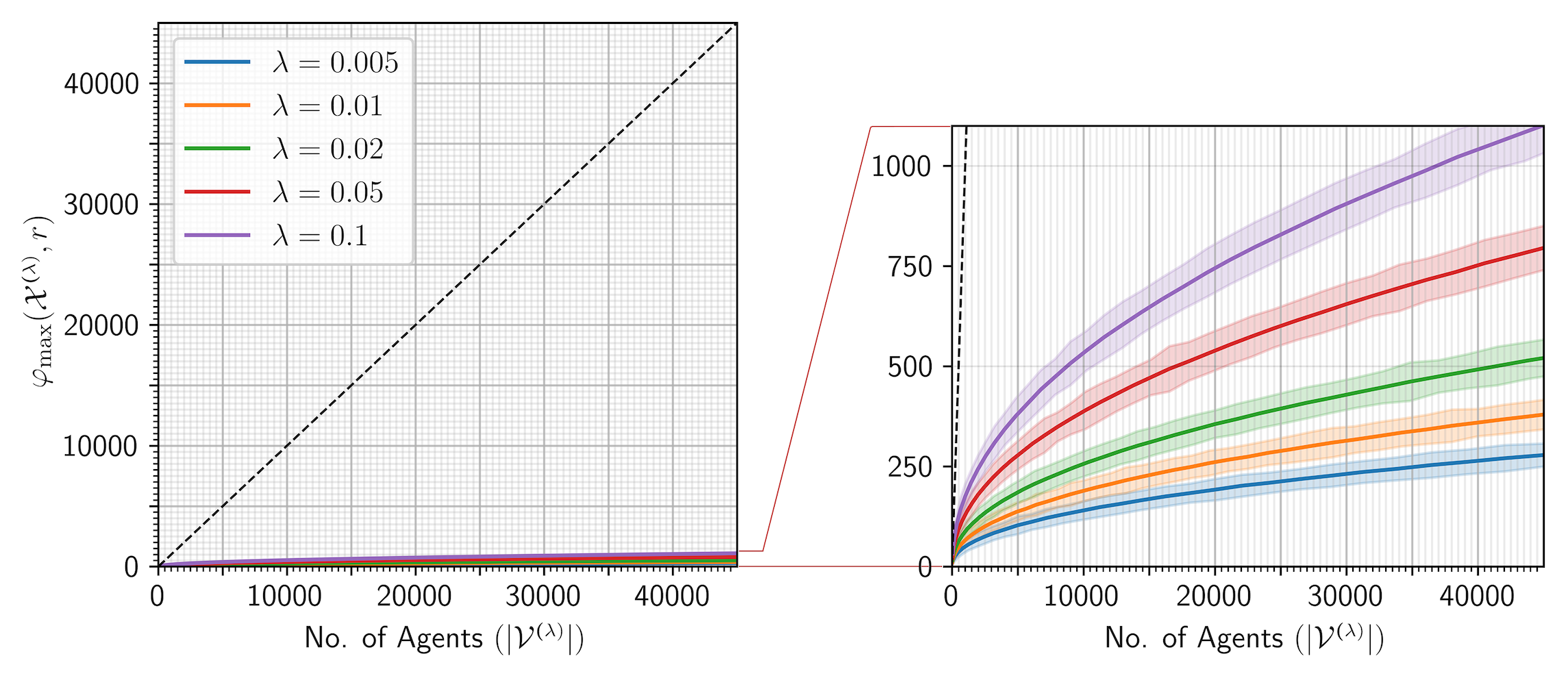}
\caption{Upper bound on the bandwidth achieved by Algorithm \ref{alg:vr} on random WSNs, computed using (\ref{eq:poisson_1d_bound}). The solid lines represent the means (conditioned on the number of agents), whereas the shaded regions represent 2 standard deviations about the means. The black dashed line shows the maximal bandwidth of the graph, which scales linearly in the number of agents.}
\label{fig:scan_statistic}
\end{figure*}

\section{Conclusions}
In this paper, a novel approach for efficient localization of large-scale wireless sensor networks (WSNs) was developed. The proposed method is based on the Extended Kalman Filter (EKF), supplemented with a novel vertex relabeling algorithm which reduces the bandwidth of the graph Laplacian, enabling one to exploit the inherent sparsity in the WSN localization problem. An upper bound on the computational complexity of the proposed algorithm was derived.
Numerical simulations were used to verify the efficacy of the proposed WSN localization approach.
We also considered the case of typical WSN configurations, modeled as random geometric graphs, showing that the vertex relabeling approach drastically reduces the computational complexity of large-scale WSN localization.

Some avenues for future research on this topic include the development of more sophisticated vertex relabeling algorithms, as well as the use of nonlinear state estimators (such as the Unscented Kalman Filter) in conjunction with sparse matrix methods.

\section*{Acknowledgements}
This research is funded by the Secure Systems Research Center (SSRC) at the Technology Innovation Institute (TII), UAE. The authors are grateful to Dr. Shreekant (Ticky) Thakkar and his team members at the SSRC for their valuable comments and support.



\bibliographystyle{plain} 
\bibliography{refs} 

\begin{thebibliography}{10}

\bibitem{aspnes2006theory}
James Aspnes, Tolga Eren, David~Kiyoshi Goldenberg, A~Stephen Morse, Walter
  Whiteley, Yang~Richard Yang, Brian~DO Anderson, and Peter~N Belhumeur.
\newblock A theory of network localization.
\newblock {\em IEEE Transactions on Mobile Computing}, 5(12):1663--1678, 2006.

\bibitem{rcm2017}
Ariful Azad, Mathias Jacquelin, Aydin Buluç, and Esmond~G. Ng.
\newblock The reverse cuthill-mckee algorithm in distributed-memory.
\newblock In {\em 2017 IEEE International Parallel and Distributed Processing
  Symposium (IPDPS)}, pages 22--31, 2017.

\bibitem{2019ekf}
Bradley~T Burchett.
\newblock Distributed kalman filters for cooperative localization of munition
  swarms.
\newblock Technical report, Rose-Hulman Institute of Technology, 2019.

\bibitem{chinn1982bandwidth}
Phyllis~Z Chinn, Jarmila Chv{\'a}talov{\'a}, Alexander~K Dewdney, and Norman~E
  Gibbs.
\newblock The bandwidth problem for graphs and matrices—a survey.
\newblock {\em Journal of Graph Theory}, 6(3):223--254, 1982.

\bibitem{cormen2022introduction}
Thomas~H Cormen, Charles~E Leiserson, Ronald~L Rivest, and Clifford Stein.
\newblock {\em Introduction to algorithms}.
\newblock MIT press, 2022.

\bibitem{daley2008poissonTheory}
Daryl~J Daley and David Vere-Jones.
\newblock {\em An Introduction to the Theory of Point Processes. Volume II:
  General Theory and Structure}.
\newblock Springer, 2008.

\bibitem{bandwidth_NPcompleteness2001Penrose}
Josep D{\i}az, Mathew~D Penrose, Jordi Petit, and Mar{\i}a Serna.
\newblock Approximating layout problems on random geometric graphs.
\newblock {\em Journal of Algorithms}, 39(1):78--116, 2001.

\bibitem{dimakis2006geographic}
Alexandros~G Dimakis, Anand~D Sarwate, and Martin~J Wainwright.
\newblock Geographic gossip: Efficient aggregation for sensor networks.
\newblock In {\em Proceedings of the 5th international conference on
  Information processing in sensor networks}, pages 69--76, 2006.

\bibitem{gilbert1961random}
Edward~N Gilbert.
\newblock Random plane networks.
\newblock {\em Journal of the society for industrial and applied mathematics},
  9(4):533--543, 1961.

\bibitem{glaz2001scan}
Joseph Glaz, Joseph~I Naus, Sylvan Wallenstein, Sylvan Wallenstein, and
  Joseph~I Naus.
\newblock {\em Scan statistics}.
\newblock Springer, 2001.

\bibitem{haenggi2012stochastic}
Martin Haenggi.
\newblock {\em Stochastic geometry for wireless networks}.
\newblock Cambridge University Press, 2012.

\bibitem{WSN2021}
Yassine Hmamouche, Mustapha Benjillali, Samir Saoudi, Halim Yanikomeroglu, and
  Marco Di~Renzo.
\newblock New trends in stochastic geometry for wireless networks: A tutorial
  and survey.
\newblock {\em Proceedings of the IEEE}, 109(7):1200--1252, 2021.

\bibitem{iet_relative_bearing}
Younghun John, Kwang-Kyo Oh, Bari{\c{s}} Fidan, and Hyo-Sung Ahn.
\newblock Fixed-time orientation estimation and network localisation of
  multi-agent systems.
\newblock {\em IET Control Theory \& Applications}, 15(1):64--76, 2021.

\bibitem{kavcic2000matrices}
Aleksandar Kavcic and Jos{\'e}~MF Moura.
\newblock Matrices with banded inverses: Inversion algorithms and factorization
  of gauss-markov processes.
\newblock {\em IEEE transactions on Information Theory}, 46(4):1495--1509,
  2000.

\bibitem{2010randomWSN}
Hichem Kenniche and Vlady Ravelomananana.
\newblock Random geometric graphs as model of wireless sensor networks.
\newblock In {\em 2010 The 2nd international conference on computer and
  automation engineering (ICCAE)}, volume~4, pages 103--107. IEEE, 2010.

\bibitem{khan2022robust}
Shiraz Khan, Inseok Hwang, and James Goppert.
\newblock Robust state estimation in the presence of stealthy cyberattacks.
\newblock In {\em 2022 American Control Conference (ACC)}, pages 304--309.
  IEEE, 2022.

\bibitem{moura2008distributing}
Usman~A Khan and Jos{\'e}~MF Moura.
\newblock Distributing the kalman filter for large-scale systems.
\newblock {\em IEEE Transactions on Signal Processing}, 56(10):4919--4935,
  2008.

\bibitem{consensus}
Wenling Li, Yingmin Jia, and Junping Du.
\newblock Distributed consensus extended kalman filter: a variance-constrained
  approach.
\newblock {\em IET Control Theory \& Applications}, 11(3):382--389, 2017.

\bibitem{liao2016improved_band_matrix_computation}
Xiangke Liao, Shengguo Li, Lizhi Cheng, and Ming Gu.
\newblock An improved divide-and-conquer algorithm for the banded matrices with
  narrow bandwidths.
\newblock {\em Computers \& Mathematics with Applications}, 71(10):1933--1943,
  2016.

\bibitem{RSS2019efficient}
Thu LN~Nguyen and Yoan Shin.
\newblock An efficient rss localization for underwater wireless sensor
  networks.
\newblock {\em Sensors}, 19(14):3105, 2019.

\bibitem{bandwidthsSurvey2014}
Liviu~Octavian Mafteiu-Scai.
\newblock The bandwidths of a matrix. a survey of algorithms.
\newblock {\em Annals of West University of Timisoara-Mathematics and Computer
  Science}, 52(2):183--223, 2014.

\bibitem{bdo_localization2007}
Guoqiang Mao, Barış Fidan, and Brian~D.O. Anderson.
\newblock Wireless sensor network localization techniques.
\newblock {\em Computer Networks}, 51(10):2529--2553, 2007.

\bibitem{embedAndProject}
Janez Povh.
\newblock On the embed and project algorithm for the graph bandwidth problem.
\newblock {\em Mathematics}, 9(17):2030, 2021.

\bibitem{2010PoissonRGG}
Yanhuai Qu, Jianan Fang, and Shuai Zhang.
\newblock Nearest neighbor nodes and connectivity of wireless sensor networks
  with poisson point process.
\newblock In {\em Proceedings of the 29th Chinese Control Conference}, pages
  4776--4780. IEEE, 2010.

\bibitem{ekf1999}
Konrad Reif, Stefan Gunther, Engin Yaz, and Rolf Unbehauen.
\newblock Stochastic stability of the discrete-time extended kalman filter.
\newblock {\em IEEE Transactions on Automatic control}, 44(4):714--728, 1999.

\bibitem{2002ekf}
S.I. Roumeliotis and G.A. Bekey.
\newblock Distributed multirobot localization.
\newblock {\em IEEE Transactions on Robotics and Automation}, 18(5):781--795,
  2002.

\bibitem{iet_loc_distance}
Jing-Ping Shao and Yu-Ping Tian.
\newblock Cooperative source localisation of multi-agent system based on a
  cooperative pe condition.
\newblock {\em IET Control Theory \& Applications}, 9(1):42--51, 2015.

\bibitem{ekf1995}
Yongkyu Song and Jessy~W Grizzle.
\newblock Asymptotic observer for discrete-time nonlinear systems.
\newblock {\em Journal of Mathematical Systems, Estimation, and Control},
  5(1):59--78, 1995.

\bibitem{ekfCRLB}
James~H. Taylor.
\newblock The cramer-rao estimation error lower bound computation for
  deterministic nonlinear systems.
\newblock In {\em 1978 IEEE Conference on Decision and Control including the
  17th Symposium on Adaptive Processes}, pages 1178--1181, 1978.

\bibitem{sdp2018}
Zengfeng Wang, Hao Zhang, Tingting Lu, and T.~Aaron Gulliver.
\newblock Cooperative rss-based localization in wireless sensor networks using
  relative error estimation and semidefinite programming.
\newblock {\em IEEE Transactions on Vehicular Technology}, 68(1):483--497,
  2019.

\bibitem{observability2015}
Ryan~K. Williams and Gaurav~S. Sukhatme.
\newblock Observability in topology-constrained multi-robot target tracking.
\newblock In {\em 2015 IEEE International Conference on Robotics and Automation
  (ICRA)}, pages 1795--1801, 2015.

\bibitem{zhao2019bearing}
Shiyu Zhao and Daniel Zelazo.
\newblock Bearing rigidity theory and its applications for control and
  estimation of network systems: Life beyond distance rigidity.
\newblock {\em IEEE Control Systems Magazine}, 39(2):66--83, 2019.

\end{thebibliography}

\end{document}